\documentclass[twocolumn,
 hf,
]{ceurart}

\usepackage[utf8]{inputenc}
\usepackage{times}
\usepackage{soul}
\usepackage{url}
\usepackage[small]{caption}
\usepackage{graphicx}
\usepackage{amsmath}
\usepackage{amsthm}
\usepackage{booktabs}
\usepackage{algorithm2e}
\usepackage{tikz}
\usetikzlibrary{fadings}
\usetikzlibrary{shapes,decorations,positioning}
\usetikzlibrary{fadings,shapes,decorations,positioning,calc}
\usetikzlibrary{decorations.markings,decorations.pathreplacing}
\usetikzlibrary{shapes.geometric,automata,positioning}
\usetikzlibrary{shadows}\usetikzlibrary{calc}
\usetikzlibrary{decorations,decorations.pathmorphing,decorations.pathreplacing}
\usetikzlibrary{calc}
\usetikzlibrary{tikzmark}

\newtheorem{theorem}{Theorem}
\newtheorem{example}[theorem]{Example}
\newtheorem{definition}[theorem]{Definition}
\newtheorem{proposition}[theorem]{Proposition}

\newtheorem{corollary}[theorem]{Corollary}

\usepackage{dlcobook}
\usepackage{xspace}
\usepackage{mathrsfs} 
\usepackage{mathtools}
\usepackage{amssymb,amsmath}
\usepackage{comment}
\usepackage{stmaryrd}
\usepackage{csquotes}
\usepackage[UKenglish]{babel}
\usepackage{enumerate}
\usepackage{microtype}
\usepackage{multicol}

\usepackage{pifont}%
\newcommand{\xmark}{\ding{55}}%

\newcommand{\minOf}[2]{\ensuremath{\min(#1,#2)}}
\renewcommand{\tuple}[1]{\ensuremath{\langle{#1}\rangle}}
\newcommand{\statesOf}[1]{\ensuremath{\pmS(#1)}}
\newcommand{\modelsOf}[1]{\ensuremath{\llbracket #1\rrbracket}}

\newcommand{\pmC}{\ensuremath{\mathbb{C}}}
\newcommand{\pmM}{\ensuremath{\mathbb{M}}}
\newcommand{\pmW}{\ensuremath{\mathbb{W}}}
\newcommand{\pmS}{\ensuremath{\mathcal{S}}}
\newcommand{\pmL}{\ensuremath{\ell}}
\newcommand{\pmP}{\ensuremath{\prec}}
\newcommand{\pmR}{\ensuremath{\mathrel{R}}}

\DeclareMathOperator{\nmableitC}{\nmableit_{\mkern-2mu\pmC}}
\DeclareMathOperator{\nmableitW}{\nmableit_{\mkern-2mu\pmW}}

\ifx\nmableit\undefined
\makeatletter
\ifx\centernot\undefined
\newcommand*{\centernot}{%
	\mathpalette\@centernot
}
\fi
\def\@centernot#1#2{%
	\mathrel{%
		\rlap{%
			\settowidth\dimen@{$\m@th#1{#2}$}%
			\kern.5\dimen@
			\settowidth\dimen@{$\m@th#1=$}%
			\kern-.5\dimen@
			$\m@th#1\not$%
		}%
		{#2}%
	}%
}
\makeatother
\DeclareRobustCommand\nmableitSymb{\mathrel{|}\joinrel\sim} %
\newcommand{\nmableit}{\nmableitSymb} %
\fi

\newcommand{\problemdef}[3]{%
\begin{center}
\begin{tabular}{l@{\hskip 0.2cm}p{5.3cm}}\toprule
\textsf{\bfseries Problem:}& #1 \\\midrule
\textsf{\bfseries Input:}& #2.\\
\textsf{\bfseries Question:}& #3?\\\bottomrule
\end{tabular}
\end{center}
} \newcommand{\allModels}[1]{\mathbb{A}_{#1}}
\newcommand{\allTeams}[1]{\mathbb{T}_{#1}}

\newif\ifhideproofs

\ifhideproofs
\usepackage{environ}
\NewEnviron{hide}{}

\fi

\newcommand{\ENT}{\problemFont{\textsc{Ent}}}

\newcommand{\succinct}{\textsc{Succ}}
\newcommand{\PDL}{\logicFont{PDL}\xspace}

\renewcommand{\PL}{\ensuremath{\mathrm{PL}}\xspace}
\newcommand{\TPL}{\ensuremath{\logicFont{TPL}}\xspace}

\newcommand{\CPL}{\ensuremath{\logicFont{CPL}}\xspace}

\newcommand{\formulaOne}{\ensuremath{\varphi}}
\newcommand{\formulaTwo}{\ensuremath{\psi}}
\newcommand{\formulaThree}{\ensuremath{\gamma}}

\newcommand{\pmLogicPref}[1]{\ensuremath{{#1}\kern-0.5ex\raise0.95ex\hbox{\tiny\logicFont{pref}}}}
\newcommand{\TPLPref}{\pmLogicPref{\TPL}}
\newcommand{\CPLPref}{\pmLogicPref{\CPL}}
\newcommand{\PDLPref}{\pmLogicPref{\PDL}}

\newcommand{\pmLogicX}[2]{\ensuremath{{#1}\kern-0.5ex\raise0.95ex\hbox{\tiny{#2}}}}

\newcommand{\pmLogicCum}[1]{\ensuremath{{#1}\kern-0.5ex\raise0.95ex\hbox{\normalfont\tiny\logicFont{cuml}}}}
\newcommand{\TPLCum}{\pmLogicCum{\TPL}}
\newcommand{\CPLCum}{\pmLogicCum{\CPL}}
\newcommand{\PDLCum}{\pmLogicCum{\PDL}}

\newcommand{\pmLogicStrongCum}[1]{\ensuremath{{#1}\kern-0.5ex\raise0.95ex\hbox{\normalfont\tiny\logicFont{cuml}[\logicFont{str}]}}}

\newcommand{\pmLogicAsym}[1]{\ensuremath{{#1}\kern-0.5ex\raise0.95ex\hbox{\normalfont\tiny\logicFont{cuml}[\logicFont{as}]}}}
\newcommand{\TPLAsym}{\pmLogicAsym{\TPL}}

\newcommand{\pmLogicSystemP}[1]{\ensuremath{{#1}\kern-0.5ex\raise0.95ex\hbox{\tiny\textbf{\logicFont{p}}}}}
\newcommand{\pmLogicSystemC}[1]{\ensuremath{{#1}\kern-0.5ex\raise0.95ex\hbox{\tiny\textbf{\logicFont{c}}}}}
\newcommand{\TPLSystemP}{\pmLogicSystemP{\TPL}}
\newcommand{\CPLSystemP}{\pmLogicSystemP{\CPL}}
\newcommand{\PDLSystemP}{\pmLogicSystemP{\PDL}}
\newcommand{\TPLSystemC}{\pmLogicSystemC{\TPL}}
\newcommand{\CPLSystemC}{\pmLogicSystemC{\CPL}}
\newcommand{\PDLSystemC}{\pmLogicSystemC{\PDL}}

\DeclareMathOperator{\leqlogm}{\leq_{m}^{\mathrm{log}}}

\newcommand{\complClFont}[1]{\mathsf{#1}}
\newcommand{\Ptime}{\complClFont{P}}

\newcommand{\NC}[1]{\complClFont{NC}^{#1}}

\newcommand{\SystemC}{\textit{System~C}}
\renewcommand{\pdl}{\ensuremath{\PL(=)}}
\renewcommand{\pdl}{\ensuremath{\PL(\logicFont{dep})}}

\newcommand{\ssLogic}{\ensuremath{\mathscr{L}}}
\newcommand{\ssSystem}{\ensuremath{\mathbb{S}}}
\newcommand{\ssFormulas}{\ensuremath{\mathcal{L}}}
\newcommand{\ssInt}{\ensuremath{\Omega}}

\usepackage{todonotes}

\newcommand{\MYParagraph}[1]{\par\smallskip\noindent\textbf{#1}}
\newcommand{\MYSubsection}[1]{\MYParagraph{#1.}}

\makeatletter
\newcommand{\textlabelmarker}[1]{%
    \protected@edef\@currentlabel{#1}%
    \phantomsection%
}
\newcommand{\textlabel}[2]{%
    \textlabelmarker{#1}%
    #1\label{#2}%
}
\makeatother
\author[1]{Kai Sauerwald}[%
]
\author[2]{Juha Kontinen}[%
]
\author[3]{Arne Meier}[%
]
\address[1]{FernUniversität in Hagen, Hagen, Germany}
\address[2]{University of Helsinki, Helsinki, Finland}
\address[3]{Leibniz Universität Hannover, Hannover, Germany}

\begin{document}
\thispagestyle{plain}
\copyrightyear{2026}
\copyrightclause{Copyright for this paper by its authors.}

\conference{16th International Workshop on Logic and Computational Complexity, July 18–19, 2026, Lisbon, Portugal}

\title{On the Complexity of Entailment for\\ Cumulative Propositional Dependence Logics}

\begin{abstract}
This paper establishes and proves complexity results for entailment for cumulative propositional dependence logic and for cumulative propositional logic with team semantics. As recently shown, cumulative logics are famously characterised by System~C and exactly captured by the cumulative models of Kraus, Lehmann and Magidor. 
This gives rise to the entailment problem via relational models, which is specifically considered here.
\end{abstract}


\maketitle

\section{Introduction}
\label{sec:introduction}
The ability to reason is one of the central features of intelligent agents and thus a major concern of artificial intelligence. 
In this paper, we study the fusion of non-monotonic reasoning with team-based reasoning.
Such a combination is interesting, as it allows reasoning in which one can express settings involving a plurality of objects (the team-based component) while also taking into account extra-logical information, such as plausibility, exceptionality, reliability, preference, or typicality (the non-monotonic component).
Notably, the combination of these approaches provides a setting that cannot be formalized by neither of the approaches alone.%

Less is known about how to construct non-monotonic entailment relations for team-based logics. Known approaches for non-monotonic propositional logics or non-monotonic first-order logics rely heavily on the properties of the underlying logics, such as, e.g., the availability of all Boolean connectives, the law of excluded middle, and presence of material implication. However, these properties are not (always) available in team-based logics.
 There are many approaches to non-monotonic logics that take inspiration from~\cite{KS_Brewka1997}, from which some are also generic, e.g., MAK models~\cite{KS_Makinson1989}, KLM-style reasoning~\cite{KS_KrausLehmannMagidor1990}, characterization logics~\cite{KS_BaumannStrass2025}, or approximation fixpoint theory~\cite{KS_DeneckerMarekTruszczynski2000}.

So far, there are only two pioneering works that consider a combination of team semantics and non-monotonic reasoning. First, the work by \citeauthor{JY23} (\citeyear{JY23}), which considers a non-monotonic team-based modal logic in the context of formal analysis of natural language. The second is by \citeauthor{KS_SauerwaldMeierKontinen2025} (\citeyear{KS_SauerwaldMeierKontinen2025}) (SMK), which also contains a broader introduction and motivation, and focuses on non-monotonic reasoning in the style of \citeauthor{KS_KrausLehmannMagidor1990} (\citeyear{KS_KrausLehmannMagidor1990}) (KLM). 

KLM (\citeyear{KS_KrausLehmannMagidor1990}) showed for classical propositional logic that cumulative entailment relations provide a stable theory of reasoning with multiple representations. Cumulative entailment relations are obtained by reasoning via the axiomatic system known as {\SystemC}. Furthermore, a cumulative entailment relation \( \nmableit \) can be represented by a cumulative model \( \pmC \), by which \( \formulaOne\nmableit\formulaTwo \) amounts (intuitively) to checking whether all minimal models of \( \formulaOne \) in \( \pmC \) are models of \( \formulaTwo \).
SMK (\citeyear{KS_SauerwaldMeierKontinen2025}) show that entailment via preferential models (which are specific cumulative models) satisfies {\SystemC} in the context of propositional dependence logic.
However, SMK provides no representation theorems for any kind of general cumulative reasoning in the context of team-based logics.

\MYParagraph{Contributions.}
This work 
considers the complexity of cumulative reasoning, which has been recently established as a stable approach in the context of team-based logics~\cite{KS_SauerwaldMeierKontinen2026}.
 Specifically, we deal with the cumulative counterparts of the following team-based logics:
\begin{itemize}
    \item Propositional logic with team-based semantics \hfill(\( \TPL \))
    \item Propositional dependence logic \hfill(\( \PDL \))
\end{itemize}
For each of these logics, we consider both reasoning via {\SystemC} and via cumulative models as in the approach by KLM.
This leads to {\SystemC}-based entailment relations for \( \PDL \), denoted by \( \PDLSystemC \). 
Also, we will consider \( \PDL \) entailment relations via cumulative models, denoted by \( \PDLCum \).
With \( \TPLCum \) and \( \TPLSystemC \), we denote the respective approaches for \( \TPL \).
We will study the complexity of cumulative entailment problems of these logics and obtain the following results:
\begin{enumerate}
    \item Entailment for $\PDLCum$ is in $\Theta^p_2$ and $\NP$-hard, while it is in $\Ptime$ and $\NC1$-hard for $\TPLCum$.
    \item Succinct Entailment for both logics is in $\Pi^p_2$ and $\Delta^p_2$-hard.
\end{enumerate}

This underlines (under reasonable complexity assumptions) subtle differences in the complexities of these problems regarding the underlying logics.

\section{Preliminaries}
\label{sec:background_team_based_logic}
In this paper, we consider propositional logics from a model-theoretic perspective.
We denote by $\Prop=\{\;p_i\mid i\in \mathbb{N}\;\}$ the countably infinite set of propositional variables. 
We consider propositional formulas in negation normal form, i.e., \PL-formulas are formed by the grammar, where $p\in\Prop$:
\[\formulaOne\Coloneqq p\mid \neg p \mid \bot\mid\top\mid \formulaOne\wedge\formulaOne\mid\formulaOne\vee\formulaOne.\]
 We write $\Prop(\formulaOne)$ for the set of variables occurring in $\formulaOne$.

\MYParagraph{Classical Propositional Logic ({\CPL}).}
For a non-empty finite subset $N\subseteq \Prop$ of propositional variables, one defines for valuations $v\colon N\to\{0,1\}$ over \( N \) and \PL-formulas \( \formulaOne \):
\[ \llbracket \formulaOne\rrbracket^c =
\{\, v\colon N\to\{0,1\}\mid v\models\formulaOne\}. \]
The valuation function $v$ is extended to the \emph{set of all \PL-formulas satisfying} $\Prop(\formulaOne)\subseteq N$, $\PL(N)$, in the usual way.
We denote by $\allModels{N}$ the set of all assignments over $N$. 
We write $\formulaOne\models^c\formulaTwo$ for $ \llbracket \formulaOne\rrbracket^c \subseteq \llbracket \formulaTwo\rrbracket^c $ and $\formulaOne\equiv^c\formulaTwo$ if both $\formulaOne\models^c\formulaTwo$ and $\formulaTwo\models^c\formulaOne$ are true.
%
%
%

%

%
%
%

\MYParagraph{Prop.\ Logic with Team Semantics ({\TPL}).}
Next, we define \emph{team semantics} for \PL-formulas (cf. \cite{HannulaKVV15,YangV16}). A \emph{team $X$} is a set of valuations for some finite $N\subseteq \Prop$. The \emph{domain $N$ of $X$} is denoted by  $\dom(X)$, and $\mathbb{T}_{N}$ denotes the set of all such teams.%

\begin{definition}[Team semantics of \PL]
    Let $X$ be a team. For any \PL-formula $\formulaOne$ with $\dom(X)\supseteq \Prop(\formulaOne)$, the \emph{satisfaction relation}, $X\models\formulaOne$, is defined inductively as:
    \begin{align*}
        &X\models p&&\text{if for all }v\in X: v\models p,\\
        &X\models\neg p&&\text{if for all } v\in X: v\not\models p,\\
        &X\models \top&&\text{is always the case},\\
        &X\models \bot&&\text{if } X=\emptyset,\\
        &X\models\formulaOne\wedge\formulaTwo&&\text{if } X\models\formulaOne \text{ and } X\models\formulaTwo,\\
        &X\models\formulaOne\vee\formulaTwo&&\text{if there exist } Y,Z\subseteq X\\
        &&&\text{s.t. }X=Y\cup Z, Y\models\formulaOne, \text{ and }Z\models\formulaTwo.
    \end{align*}
\end{definition}
The set of all teams \( X \) with \( X\models\formulaOne \) is denoted by \( \modelsOf{\formulaOne}^t \).
For any two \PL-formulas \( \formulaOne,\formulaTwo \), we write $\formulaOne\models^t\formulaTwo $ if \( \modelsOf{\formulaOne}^t \subseteq \modelsOf{\formulaTwo}^t \).
Write $\formulaOne\equiv^t\formulaTwo$ if both $\formulaOne\models^t\formulaTwo$ and $\formulaTwo\models^t\formulaOne$ are true.
We define the following properties for a formula~$\formulaOne$:
\begin{align*}
    X&\models \formulaOne \iff \text{for all } v\in X,~\{v\}\models\formulaOne &{\small(\textbf{Flatness})}\\
    \emptyset &\models \formulaOne &{(\textbf{Empty team})}\\
    X &\models \formulaOne \text{ and } Y\,{\subseteq}\, X, \text{ then }{Y\models \formulaOne} &\hspace*{-2em}{\small(\textbf{Downward closure})}
\end{align*}
\begin{proposition}\label{prop:tpl_pincl_properties}
    {\TPL} has the properties flatness, empty team, and downward closure. 
\end{proposition}

Due to the flatness property, logical entailment of propositional logic with team-based semantics $\models^t$ and %
classical semantics $\models^c$ coincide. 
\MYParagraph{Propositional Dependence Logic ({\PDL}).}
\label{sec:pdl_pincl}
A \emph{(propositional) dependence atom} is a string %
$\dep{\vec{a},b}$, 
in which $\vec{a}=a_1,\dots,a_k$ and $b$ are propositional variables from \Prop. 
A team $X$ \emph{satisfies a dependence atom}, $X \models \dep{\vec{a},b}$, if for all $v, v' \in X$, $v(\vec{a})=v'(\vec{a})$ implies $v(b)=v'(b)$.
A dependence atom where the first component is empty will be abbreviated as $\dep{p}$ and called a \emph{constancy atom}. %
The language of \emph{propositional dependence logic} (\pdl) is defined as \PL-formulas extended by dependence atoms.
\begin{example}\label{example_dep_atm_propositional}
    Consider the team $X$ over $\{p,q,r\}$ defined by:\smallskip

        \noindent\begin{minipage}[c]{0.35\linewidth}
        \begin{tabular}{cccc}
            \toprule
            &$p$&$q$&$r$\\\midrule
            $v_1$&$1$&$0$&$0$\\
            $v_2$&$0$&$1$&$0$\\
            $v_3$&$0$&$1$&$0$\\
            \bottomrule
        \end{tabular}
        \end{minipage}
    \begin{minipage}[c]{0.64\linewidth}
    Here, we have that $X\models\dep{p,q}$ and $X\models\dep{r}$. 
    Moreover, we have that $X\models\dep{p}\vee \dep{p}$, however it is true that $X\not \models\dep{p}$ as the value of $p$ is not overall constant.
    \end{minipage}
\end{example}

%
%
\begin{proposition}\label{prop:pdl_pincl_properties}
    \PDL has the empty team and the downward closure property, but not the flatness property. 
\end{proposition}

\MYParagraph{Generic View on Logics.} 
Some parts of this paper require a generic perspective on logics, which we discuss next.
A \emph{satisfaction system} is a triple \( \ssSystem = \tuple{\ssFormulas,\ssInt,\models} \), where \( \ssFormulas \) is the set of \emph{formulas}, \( \ssInt \) is the set of \emph{interpretations}, and  \( {\models} \subseteq \ssInt \times \ssFormulas \) is the \emph{satisfaction relation}. 
We write $\modelsOf{\formulaOne}^\ssSystem=\{ \omega \in \ssInt \mid \omega \models \formulaOne \} $ for the set of all models of the formula $\alpha\in\ssFormulas$.
An \emph{entailment relation} for a satisfaction system is a relation \( {\Vdash} \subseteq \ssFormulas \times \ssFormulas\) such that %
$\alpha \Vdash \gamma$ if and only if $\beta \Vdash \gamma$, 
whenever $\modelsOf{\formulaOne}^\ssSystem=\modelsOf{\formulaTwo}^\ssSystem $.
A satisfaction system \( \ssSystem \) together with an entailment relation \( \Vdash \) is called a \emph{logic} and  denoted by
\( \ssLogic = \tuple{\ssFormulas,\ssInt,\models,\Vdash} \).
The propositional logics discussed in this section fit into this general model-theoretic view for each \( N \subseteq \Prop \) as follows:
\begin{center}
    \begin{tabular}{@{}l@{\hskip0.05em}l@{}}
        \( \CPL_{N} \)  
        & \( {=}\, \tuple{\PL(N),\allModels{N},\models,\models^c} \)\\
        \( \TPL_{N} \) 
        & \( {=}\, \tuple{\PL(N),\allTeams{N},\models,\models^t} \)\\
        \( \PDL_{N} \) 
        & \( {=}\, \tuple{\pdlN, \allTeams{\!N},\models,\models^t}\)
    \end{tabular}
\end{center}
When there is no ambiguity, we will write $\models$ instead of $\models^t$.
Moreover, we will use \( \CPL \) to denote the class consisting of all logics \( \CPL_{N} \) for any non-empty finite subset $N\subseteq \Prop$.
The classes \( \TPL \) and \( \PDL \) are defined analogous.

\section{Cumulative Reasoning}
\label{sec:background_klmstyle}
We consider the construction of  entailment relations.

\MYSubsection{{\SystemC}}
 We make use of the following rules for calculi:
\begin{center}
	\vspace{-1em}
	\begin{minipage}[b]{0.51\linewidth}
		\begin{align}
			&\frac{\formulaOne\equiv\formulaTwo\hspace{0.5cm}\formulaOne\nmableit\formulaThree}{\formulaTwo\nmableit\formulaThree} \tag{LLE}\label{pstl:LLE}\\[0.25em]
			&\frac{\formulaOne\land\formulaTwo\nmableit\formulaThree\hspace{0.5cm}\formulaOne\nmableit\formulaTwo}{\formulaOne\nmableit\formulaThree} \tag{Cut}\label{pstl:Cut}
		\end{align}
	\end{minipage}\begin{minipage}[b]{0.45\linewidth}
		\begin{align}
			&\frac{\formulaOne\models\formulaTwo\hspace{0.5cm}\formulaThree\nmableit\formulaOne}{\formulaThree\nmableit\formulaTwo} \tag{RW}\label{pstl:RW}\\[0.25em]
			&\frac{\formulaOne\nmableit\formulaTwo\hspace{0.5cm}\formulaOne\nmableit\formulaThree}{\formulaOne\land\formulaTwo\nmableit\formulaThree} \tag{CM}\label{pstl:CM}
		\end{align}
	\end{minipage}
\end{center}
Note that \( \models \) is a placeholder for the entailment relation \( \Vdash \) of an underlying logic \( \ssLogic = \tuple{\ssFormulas,\ssInt,\models,\Vdash}  \), and \( \equiv \) is the respective semantic equivalence from $\ssLogic$.
{\SystemC}  consists of the rules  \eqref{pstl:RW}, \eqref{pstl:LLE},  \eqref{pstl:CM}, \eqref{pstl:Cut} and reflexivity  (\textlabel{Ref}{pstl:Ref}) $\formulaOne \nmableit \formulaOne$ for all formulas \cite{KS_KrausLehmannMagidor1990,KS_Gabbay1984}.
We say that an entailment relation \( \nmableit \) \emph{satisfies} {\SystemC} if \( \nmableit \) is closed under all rules of {\SystemC}.

\MYSubsection{Relational Models}
For a relation \( {\pmR} \subseteq \mathcal{S} \times \mathcal{S} \) on a set \( \mathcal{S} \) and a subset \( S \subseteq \mathcal{S} \), an element \( s \in S \) is called \emph{minimal in \( S \) with respect to \( \pmR \)} if for each \( s' \in S \) holds  \( \lnot(s' \pmR s) \). 
Then, \( \minOf{S}{{\pmR}} \)  is the set of all \( s \in S \)  that are minimal in \( S \) with respect to \( \pmR \).
Moreover, for a set of interpretations \( M \) and a formula \( \varphi \) of  \( \ssLogic = \tuple{\ssFormulas,\ssInt,\models,\Vdash}  \), we write \( M \models \varphi \) if for all \( \omega\in M \) we have that \( \omega \models \varphi \).
\begin{definition}[Shoman~\citeyear{KS_Shoham1988}, Dix and Makinson~\citeyear{KS_Makinson1989,KS_DixMakinson1992}]
	Let \( \ssLogic = \tuple{\ssFormulas,\ssInt,\models,\Vdash}  \) be a logic.
	A \emph{relational model} for \(  \ssLogic  \) is a triple \( \pmM=\tuple{\pmS,\pmL,\pmR} \) where \( \pmS \) is a set, \( \pmL \colon \pmS \to \mathcal{P}(\Omega) \), and \( \pmR \) is a binary relation on \( \pmS \).
\end{definition}
We say \( S \subseteq \pmS \) is \emph{smooth} if for each \( s \in S \), we either have that \( s\in \minOf{S}{\pmR} \), or there exists a state \( s'\in \minOf{S}{\pmR} \) with \( s' \pmR s \).
For \( \formulaOne \in \ssFormulas \) and \( s\in\pmS \), we denote by \( \statesOf{\formulaOne}=\{\, s \in \pmS \mid \ell(s) \models \formulaOne\, \} \) the set of states that satisfy \( \formulaOne \). %
With \( \minOf{\modelsOf{\formulaOne}}{\pmR} = \bigcup\{\, \pmL(s) \mid s\in \minOf{\statesOf{\formulaOne}}{\pmR}  \,\} \), we denote the set of all interpretations that appear in \( \pmL(s) \) for any minimal state \( s \) that satisfy \( \formulaOne \).
We will deal, in this paper, with specific types of relational models, which we define in the following.

\begin{definition}
	A relational model \( \pmC=\tuple{\pmS,\pmL,\pmR} \) for a logic \( \ssLogic = \tuple{\ssFormulas,\ssInt,\models,\Vdash}  \) is called 
    \emph{cumulative} if for all \( \formulaOne \in \ssFormulas \) the set \( \statesOf{\formulaOne} \) is smooth.
    A cumulative model \( \pmC \) is \emph{strong} if \( \pmR \) is asymmetric ($x R y$ implies $\lnot(y R x)$) and \( \minOf{\statesOf{\formulaOne}}{\pmR} \) has exactly one element for each formula \( \formulaOne \in \ssFormulas \).
\end{definition}
Because cumulative models satisfy smoothness, they avoid the problem of reasoning via arbitrary relational models, in which the set $\minOf{\modelsOf{\formulaOne}}{\pmR}$ might be empty, even when $\statesOf{\formulaOne}$ is non-empty. 
Entailment relations, which are induced by a relational model, are defined as follows.
\begin{definition}
	Let \( \ssLogic = \tuple{\ssFormulas,\ssInt,\models,\Vdash}  \) be a logic.
	The \emph{entailment relation} \( {\nmableit_{\pmM}} \subseteq \ssFormulas \times \ssFormulas \) for a relational model \( \pmM=\tuple{\pmS,\pmL,\pmR} \) for \(  \ssLogic \) is given by
	\begin{equation*}
		\formulaOne \nmableit_{\pmM} \formulaTwo \text{ if }  \minOf{\modelsOf{\formulaOne}}{\pmR} \subseteq \modelsOf{\formulaTwo}\ .
	\end{equation*}
\end{definition}
An entailment relation \( {\nmableit} \subseteq \ssFormulas \times \ssFormulas \) is called \emph{(strongly) cumulative} if there is a (strong) cumulative model \( \pmC \) for \( \ssLogic \) such that  \( {\nmableit} =   {\nmableitC}  \). 

\MYSubsection{Classes of Entailment Relations}
If \( \logicFont{L} \) is a class of logics (such as \PDL, \TPL, or \CPL),  we denote with \( \pmLogicSystemC{\logicFont{L}} \)  the class of all entailment relations \( {\nmableit} \) for \( \tuple{\ssFormulas,\ssInt,\models,\Vdash} \in \logicFont{L} \) that satisfy {\SystemC}. 
Similarly, we use  \( \pmLogicCum{\logicFont{L}} \) (\( \pmLogicStrongCum{\logicFont{L}} \)) for the class of all (strong) cumulative entailment relations.

\section{Complexity of Cumulative Entailments}
\label{sec:complexity_cuml}

\begin{table}[tb]
\centering
\begin{tabular}{l@{}c@{}cl}
     \toprule
     \textbf{Problem}                             & \textbf{Tractable} &       \textbf{Complexity}       & \textbf{Result}                         \\ \midrule
     $\ENT(\PDLCum)$                &  \xmark  &   $\in\Theta_2^p$, $\NP$-hard   & Thm.~\ref{thm:pdlcuml}                 \\
     $\succinct\ENT(\PDLCum)$       &  \xmark   & $\in\Pi^p_2$, $\Delta^p_2$-hard & Thm.~\ref{thm:succ-PDLcuml} \\
     $\ENT(\TPLCum)$               &     \checkmark      &   $\in\Ptime$, $\NC{1}$-hard    & Col.~\ref{cor:tplcuml}                 \\
     $\succinct\ENT(\TPLCum)$      &  \xmark   & $\in\Pi^p_2$, $\Delta^p_2$-hard & Col.~\ref{cor:tplcuml}                 \\ \bottomrule
 \end{tabular}
\caption{Overview of novel complexity results for entailment based on cumulative models for $\PDL$ and $\TPL$. 
     $\ENT$/$\succinct\ENT$ are the (succinct) entailment problem.}
\label{tab:compl-overview}
\end{table}

In this section, we study the complexity of entailment for cumulative preferential logics. We consider both the standard and the succinct version of the problem, where the latter uses a circuit representation of the knowledge base.
Recall, that a relational model \( \pmC=\tuple{\pmS,\pmL,\pmR} \) for a logic \( \ssLogic = \tuple{\ssFormulas,\ssInt,\models,\models^{\ssLogic}}  \) is cumulative if for every formula \(\varphi\in\ssFormulas\) the set \( \statesOf{\formulaOne} \) is \emph{smooth}, i.e., for each \( s \in \statesOf{\formulaOne} \) we have either	\( s\in \minOf{\statesOf{\formulaOne}}{\pmR} \) or there exists a state \( s'\in \minOf{\statesOf{\formulaOne}}{\pmR} \) with \( s' \pmR s \). 
In the following, we will turn towards the central decision problems. 
We start with the standard entailment problem for cumulative propositional logic.
\problemdef{$\ENT(\PDLCum)$ --- entailment problem for cumulative propositional dependence logic}{A finite cumulative model $\pmW = \tuple{\pmS,\pmL,\pmP}$ for $\PDL$ and $\formulaOne,\formulaTwo\in\PL$}{Is it true that $\formulaOne\nmableitW\formulaTwo$}

\begin{theorem}
     $\ENT(\PDLCum)$ is in $\Theta^p_2$ and $\NP$-hard.\label{thm:pdlcuml}
\end{theorem}
\begin{proof}
     For membership, we make use of the cumulative input model \(\pmC=\tuple{\pmS,\pmL,\pmR}\). 
     Regarding all \( s\in \pmS(\varphi)\) we first need to compute $\statesOf{\formulaOne}$ which is the set of all formulas such that $\ell(s)\models\varphi$. 
     As the model checking problem for $\PDL$ is $\NP$-complete \cite[Thm.~1]{DBLP:conf/sofsem/EbbingL12}, we need to go through all states \( s\in \pmS \) and check whether \( \ell(s)\models\varphi \) holds.
     This can be done with an $\NP$ oracle.
     Next, we need to compute the minimal states $\minOf{\statesOf{\formulaOne}}{\pmR}$. 
     This can be done by checking for each state \( s\in \statesOf{\formulaOne} \) whether there is a state \( s' \in \statesOf{\formulaOne} \) such that \( s' \pmR s \) holds.
     This can be done in polynomial time in the size of the model.
     To identify the minimal states, we only need to identify the sinks in the graph induced by \( \pmR \) on \( \statesOf{\formulaOne} \).
     Finally, we need to check whether for all minimal states \( s\in \minOf{\statesOf{\formulaOne}}{\pmR} \) it holds that \( \ell(s)\models\formulaTwo \). 
     Again, this can be done with an $\NP$ oracle.
     Thus, overall we can decide $\formulaOne\nmableitW\formulaTwo$ in $\Theta^p_2$. 
     Notice that this upper bound is in line with the upper bound for the preferential logic $\PDLPref$, which has to be considered as a special case of cumulative models.
     
     The lower bound follows from the underlying classical complexity of model checking for $\PDL$ by constructing a simple cumulative model \( (T,\formulaOne)\mapsto ((\{T\},\mathrm{id}_\pmS,\emptyset),\top,\formulaOne) \) which is already preferential.
\end{proof}

For preferential models circuit families of size $(2^{n})^{O(1)}$ for $\succinct\ENT(\PDLPref)$ have been considered~\cite{KS_SauerwaldMeierKontinen2025}. 
This was correct, since in such models, a state maps to a team. 
For cumulative models, a state maps to a set of teams, requiring a further exponential jump in the representation, i.e., we need to consider circuit families of size $(2^{2^n})^{O(1)}$ for $\succinct\ENT(\PDLCum)$.
To that end, we will next define an adequate notion that appropriately addresses these large models.

\begin{definition}
    Let $N\subseteq\Prop$ be a set of propositions with $|N|=n$, $\pmS=\{0,1\}^m$ be a set for $m\in (2^{2^n})^{O(1)}$, and ${\pmP}\subseteq\pmS\times\pmS$ be a strict partial order. 
    Now let $\pmW$ be a cumulative model $\pmC=\tuple{\pmS,\pmL,\pmP}$ such that $\pmL\colon\pmS\rightharpoonup\mathcal{P}(\allTeams{N})$ is a partial labelling. 
    Let there be two $n^{O(1)}$-sized circuit families $\mathcal L,\mathcal O$ (labelling, ordering) such that:
    \begin{enumerate}
        \item $\ell$ is computed by $\mathcal L$,
        \item $\mathcal O\colon\pmS\times\pmS\rightharpoonup\{0,1\}$ is a partial function such that for $s,s'\in \pmS$, the circuit outputs $1$ if and only if $s\pmP s'$ is true.
    \end{enumerate}
    We call $(\mathcal L,\mathcal O)$ an \emph{$n^{O(1)}$-sized circuit representation} of $\pmW$.
\end{definition}
We can derive from this definition the succinct entailment problem for cumulative propositional dependence logic, $\succinct\ENT(\PDLCum)$, which is the problem considering only instances that have a $n^{O(1)}$-sized circuit representation of the preferential model, i.e., the input then is of the form $\tuple{\tuple{\mathcal O,\mathcal L},\formulaOne,\formulaTwo}$.

\begin{theorem}
    $\succinct\ENT(\PDLCum)$ is in $\Pi^p_2$ and $\Delta^p_2$-hard 
    under
    $\leqlogm$-reductions.\label{thm:succ-PDLcuml}
\end{theorem}

\paragraph{Entailment for cumulative preferential logics with team semantics.} 
We derive the natural entailment problem for cumulative propositional logic with team semantics from the one for cumulative propositional dependence logic. 

\problemdef{$\ENT(\TPLCum)$ --- entailment problem for cumulative propositional logic with team-based semantics}{A finite cumulative model $\pmW = \tuple{\pmS,\pmL,\pmP}$ for \TPL and $\formulaOne,\formulaTwo\in\PL$}{Is it true that $\formulaOne\nmableitW\formulaTwo$}

This results in a better complexity for the non-succinct version.

\pagebreak[3]%
\begin{theorem} The following holds:
    \begin{enumerate}
        \item $\ENT(\TPLCum)\in\Ptime$ and is $\NC{1}$-hard.
        \item $\succinct\ENT(\TPLCum)\in\Pi_2^p$ and is $\Delta_2^p$-hard under $\leqlogm$-reductions.
    \end{enumerate}\label{cor:tplcuml}
\end{theorem}%
\begin{proof}
     Regarding the first item, similarly as for preferential propositional logic with team semantics, the simpler model checking problem for the underlying team logic $\TPL$---which is in $\Ptime$ by \cite[Tab.~1]{DBLP:conf/sofsem/EbbingL12}---the oracle calls therefore become easier and can be computed on-the-fly in the $\Ptime$-machine. 
     For the hardness, the lower bound from $\NC{1}$ (polynomial sized circuits of logarithmic depth with bounded fan-in; it is contained in nondeterministic logarithmic space) for classical propositional model checking~\cite{DBLP:journals/siamcomp/BussCGR92} applies.

     For item two, we can use $\forall\exists$-nondeterminism to check if there is no minimal state $s$ that violates the entailment.
     That is, the label of $s$ satisfies $\varphi$ but not $\psi$.
     Regarding the lower bound, we can use the $\Delta_2^p$-hardness of $\succinct\ENT(\PDLCum)$ and the fact that $\TPLPref\subseteq\TPLCum$ to obtain a $\Delta_2^p$-lower bound for $\succinct\ENT(\TPLCum)$.
\end{proof} 

\section{Conclusion}
\label{sec:conclusion}

In this paper, we proved complexity results of entailment problems based on cumulative reasoning for propositional dependence logic and propositional logic with team semantics.
These cumulative logics show a diverse landscape as shown in  Figure~\ref{fig:landscape} and now are better understood with respect to their computational complexity.
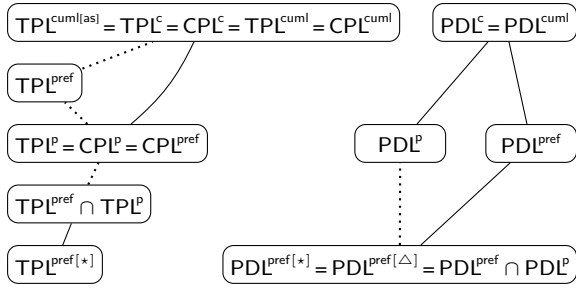
\begin{figure}[t]
\small
    \begin{tikzpicture}
    \def\mydist{0.8cm}
    \def\mydistP{2*\mydist}
    \def\myheight{0.5cm}
        \begin{scope}
        \node[anchor=west,minimum height=\myheight,draw,rectangle,rounded corners] (TPLPrefStar) at (0,0) {\( \pmLogicX{\TPL}{\logicFont{pref}\ensuremath{[\star]}} \)};
        \node[anchor=west,minimum height=\myheight,draw,rectangle,rounded corners] (TPLPrefCapP) at ([yshift=\mydist]TPLPrefStar.west) {\( \TPLPref \cap \pmLogicX{\TPL}{\logicFont{p}} \)};
        
        \node[anchor=west,minimum height=\myheight,draw,rectangle,rounded corners] (TPLp) at ([yshift=\mydist]TPLPrefCapP.west) {\( \TPLSystemP\)\,{=}\,\( \CPLSystemP \)\,{=}\,\( \CPLPref \)};
        \node[anchor=west,minimum height=\myheight,draw,rectangle,rounded corners] (TPLPref) at ([yshift=\mydist]TPLp.west) {\( \TPLPref \)};
        \node[anchor=west,minimum height=\myheight,draw,rectangle,rounded corners] (TPLc) at ([yshift=\mydist]TPLPref.west) {\( \TPLAsym \)\,{=}\,\( \TPLSystemC \)\,{=}\,\( \CPLSystemC \)\,{=}\,\( \TPLCum \)\,{=}\,\( \CPLCum \)};

        \draw (TPLPrefStar) edge [] (TPLPrefCapP);
        \draw (TPLPrefCapP) edge [dotted,thick] (TPLp);
        \draw (TPLp) edge [dotted,thick] (TPLPref);
        \draw (TPLp) edge [bend right=12] (TPLc);
        \draw (TPLPref) edge [dotted,thick] (TPLc);
        \end{scope}
        \begin{scope}
            \node[anchor=east,minimum height=\myheight,draw,rectangle,rounded corners] (PDLPrefStar) at (\columnwidth,0) {\( \pmLogicX{\PDL}{\logicFont{pref}\ensuremath{[\star]}} \)\,{=}\,\( \pmLogicX{\PDL}{\logicFont{pref}\ensuremath{[\triangle]}} \)\,{=}\,\( \PDLPref \cap \PDLSystemP \)};
        \node[anchor=east,minimum height=\myheight,draw,rectangle,rounded corners,minimum width=1.2cm] (PDLp) at ([xshift=-6em,yshift=\mydistP]PDLPrefStar.east) {\( \PDLSystemP \)};
        
        \node[anchor=east,minimum height=\myheight,draw,rectangle,rounded corners,minimum width=1.2cm] (PDLPref) at ([yshift=\mydistP]PDLPrefStar.east) {\( \PDLPref \)};
        \node[anchor=east,minimum height=\myheight,draw,rectangle,rounded corners] (PDLc) at ([yshift=\mydistP]PDLPref.east) {\( \PDLSystemC \)\,{=}\,\( \PDLCum \)};

        \draw (PDLPrefStar) edge [dotted,thick] (PDLp);
        \draw (PDLPrefStar) edge [] (PDLPref);
        \draw (PDLp) edge [] (PDLc);
        \draw (PDLPref) edge [] (PDLc);
        \end{scope}
    \end{tikzpicture}
    \vspace{-0.5cm}
    \caption{Landscape of classes of entailment relations. 
    Solid lines denote strict inclusions, and dotted lines denote inclusions, where it is unknown whether the inclusion is strict.
    For not defined classes, consult  SMK~({\protect\citeyear{KS_SauerwaldMeierKontinen2025}}).
    }
    \vspace{-1.5em}
    \label{fig:landscape}
\end{figure}

On our future agenda, we plan to find matching upper and lower complexity bounds yielding completeness results. 
Also logics, such as preferential reasoning (KLM \citeyear{KS_KrausLehmannMagidor1990}), c-inference~\cite{KS_Kern-Isberner2001}, lexicographic inference~\cite{KS_Lehmann1995a} or System~W~\cite{KS_KomoBeierle2022}, are interesting regarding complexity.

\bibliography{merged.bib}

\end{document}